\documentclass[copyright, creativecommons]{eptcs}
\providecommand{\event}{AFL 2014} % Name of the event you are submitting to
\usepackage{amsmath}
\usepackage[english]{babel}
\usepackage{graphicx}
\usepackage{epstopdf}
\usepackage{amsfonts}
\usepackage{amssymb}
\usepackage{subfigure}

\newcommand{\qed}{$\square$}
\newcommand{\cBC}{C_\mathrm{BC}}

\newtheorem{theorem}{Theorem}[section]
\newtheorem{corollary}[theorem]{Corollary}

\newenvironment{proof}[1][Proof]{\begin{trivlist}
\item[\hskip \labelsep {\bfseries #1}]}{\end{trivlist}}
\newenvironment{definition}[1][Definition]{\begin{trivlist}
\item[\hskip \labelsep {\bfseries #1}]}{\end{trivlist}}

\title{Boolean Circuit Complexity of Regular Languages}
\author{Maris Valdats
\institute{University of Latvia\\
Faculty of Computing\\
Riga, Rai\c{n}a Bulv. 19, Latvia\\}
\email{d20416@lanet.lv}}

\def\titlerunning{Boolean Circuit Complexity of Regular Languages}
\def\authorrunning{M. Valdats}
\begin{document}
\maketitle

\begin{abstract}
In this paper we define a new descriptional complexity measure for
Deterministic Finite Automata, BC-complexity, as an alternative to the state complexity.
We prove that for two DFAs with the same number of states BC-complexity can differ exponentially.
In some cases minimization of DFA can lead to an exponential increase in BC-complexity,
on the other hand BC-complexity of DFAs
with a large state space which are obtained by some standard constructions (determinization of NFA, language operations),
is reasonably small.
But our main result is the analogue of the "Shannon effect" for finite automata:
almost all DFAs with a fixed number of states have BC-complexity that is close to the maximum.
\end{abstract}
\label{1nod}
 State complexity of deterministic finite automata (DFA)~\cite{B70}\cite{H79} has been analyzed for more than
50 years and all this time has been
the main measure to estimate the descriptional complexity of finite automata. Minimization algorithm~\cite{H71}
for it was developed as well as methods to prove upper and lower bounds for various languages.

It is hard to
find any evidence of another complexity measure for finite automata.
Transition complexity~\cite{G05} could be one, it counts the number
of transitions, but there is not much use of it for DFAs (it is proportional to the state complexity),
it is used in the nondeterministic case.

But intuitively not all DFAs with the same number of states have the same complexity.
We try to illustrate it with the following example.

Consider a DFA that recognizes a language in the binary alphabet
which consists of words in which there is an even number of ones
among the last 1000 input letters. One can easily prove that it needs $2^{1000}$ states,
however such a DFA can easily be implemented by keeping its
state space in a 1000 bit register which remembers the last 1000 input letters.

On the other hand, consider a "random" DFA with a binary input tape and $2^{1000}$ states. There is essentially no better way
to describe it as with its state transition table which consists of $2^{1001}$ lines which (as it is widely
assumed) is more than particles in our universe.

%Although the state of this automaton can also be represented as a 1000 bit register, 
%transition function represented in any reasonable way (e.g. as a Boolean circuit) would
%still have a complexity of order $2^{1000}$.

It is easy to represent a large number of states in a compact form: $2^n$ states fit into $n$ state bits
of the state register. This is true for the "random" DFA as well. But the computation performed by the transition function
on this register can be very easy in some cases and hard in some other.
Therefore it seems natural to introduce a complexity measure for DFAs which
measures the complexity of the transition function. 

Automata with a large state space which is kept in a state register have been used before, but
not in the widest sense. One example of such a usage is FAPKC~\cite{TR09}
(Finite Automata Public Key Cryptosystem), a public key cryptosystem
developed in the 80's by Renji Tao. In FAPKC the state space of an automaton is
considered to be a vector space and the transition function is expressed as a polynomial over a finite field.

In this paper we consider the following model: we arbitrarily encode the state space into
a bit vector (state register) and express the transition function as a Boolean circuit.
The BC-complexity of the DFA is (approximately) the complexity of this circuit and this notion extends to regular
languages in a natural way.

BC-complexity was first analyzed in~\cite{V11} where it was considered for transducers. Here we define
it for DFAs what allows to extend the definition to regular languages.

The main result of this paper is the Shannon effect for the BC-complexity of regular languages:
it turns out that most of the languages have BC-complexity that is close to the maximum. To obtain it we first
estimate upper (and lower) bounds for BC-complexity compared to state complexity (Theorem.~\ref{apaksRobMinim}),
afterwards by counting argument we show that the complexity of most of the languages is around this upper bound.

Influence of state minimization to BC-complexity were analyzed already in~\cite{V11} for transducers and for DFAs
it is essentially the same: it turns out that for some regular languages BC-complexity of their minimal automaton
is much (superpolynomially) larger than BC-complexity for some other (non-minimal) DFA that recognizes it.
Finally we look how BC-complexity behaves if we do some standard constructions on automata
(determinization of an NFA, language operations).

%The paper is structured as follows. Section~\ref{2nod} discusses all the necessary prerequisites,
%Section~\ref{3nod} defines the representation of automaton with a Boolean circuit, Section~\ref{4nod}
%defines the main concept --- BC-complexity of an automaton and discusses its relation
%to state complexity. Section~\ref{6nod} is devoted to an interesting fact that minimizing state complexity
%can (superpolynomially) increase BC-complexity and Section~\ref{7nod} concludes the work with
%a discussion on a possible further research.

\section{Preliminaries}
\label{2nod}
\subsection{Finite Automata and Regular Languages}
We use a standard notion of DFA~\cite{M55}, it is a tuple $(Q, \Sigma, \delta, q_0, \tilde{Q})$, where
$Q$ is the state space,
$\Sigma$ is the input alphabet,
$\delta: \Sigma\times Q\rightarrow Q$ is the transition function,
$q_0\in Q$ is the start state and
$\tilde{Q}\subseteq Q$ is the set of accepting states.

%A finite deterministic automaton (DFA)~\cite{M55}  is a tuple $(Q, \Sigma, \delta, q_0, \tilde{Q})$, where
%\begin{enumerate}
%  \item $Q$ is the state space (a finite set)
%  \item $\Sigma$ is the input alphabet (a finite set)
%  \item $\delta: \Sigma\times Q\rightarrow Q$ is the transition function
%  \item $q_0\in Q$ is the start state
%  \item $\tilde{Q}\subseteq Q$ is a set of accepting states
%\end{enumerate}

DFA starts computation
in the state $q_0$ and in each step it reads an input letter $x\in \Sigma$  and changes its state.
If the current state of a DFA is $q\in Q$ and it reads an input letter $x\in \Sigma$ then it
moves to state $\delta(x, q)$. If after reading the input word DFA
is in a state $q\in \tilde{Q}$ then this word is accepted, otherwise it is rejected. DFA $A$ recognizes
language $L$ iff it accepts all words from this language and rejects all words not in the language.
%$\omega\in L\iff A \mbox{ accepts } \omega$.
Two DFAs that recognize the same language are called equivalent.
 
The state complexity of a DFA is the number of states in its state space $C_s(A)=|Q|$. For each DFA $A$
there is a unique minimal DFA $M(A)$ which is equivalent to $A$ and has minimal state complexity.
There is an effective minimization algorithm for finding it~\cite{B70}.

We will need the estimation of the number of DFAs with $s$ states.
Denote $\mathfrak{A}_s$ to be
the number of pairwise non-equivalent minimal DFAs with $s$ states over $k$-letter alphabet.
In \cite{S01} it is estimated to be larger
than $2^{s-1}(s-1)s^{(k-1)s}$, we will use the following reduced estimation (true for $s\geq 3$) which will be sufficient for us:
\begin{theorem}[\cite{S01}]
\label{autSkaits}
$\mathfrak{A}_s\geq 2^{s}s^{(k-1)s}$ for $s\geq 3$.
\end{theorem} 

\subsection{Boolean circuits}

%We will use Boolean circuits to represent transition function of DFA as well as the characteristic function of $\tilde{Q}$.
%Here we define them
%and discuss some basic properties. We will restrict our attention to those circuits that are in the
%standard base ($\&$, $\vee$, $\neg$).
We will use the standard notion of a Boolean circuit
and restrict our attention to circuits in the standard base ($\&$, $\vee$, $\neg$).
%A Boolean circuit is a directed acyclic graph with indegree for any vertex at most 2. The nodes of indegree 0 are called inputs, and are labeled with a variable $x_i$
%or with a constant 0 or 1. The nodes of indegree $k\in\{1, 2\}$ are called gates and are labeled with Boolean functions:
%\begin{enumerate}
%\item $\neg$ if $k=1$
%\item $\&$ or $\vee$ if $k=2$
%\end{enumerate}
% Some of the nodes are designated as output nodes and are labeled with output letters $y_1,\dots, y_m$.
The size of the circuit $C(F)$ is
the number of gates plus the number of outputs of the circuit $F$.
Boolean circuit $F$ with $n$ inputs and $m$ outputs represents a Boolean function $(y_1,\dots, y_m)=F(x_1,\dots, x_n)$ in a natural way.
%Every node
%computes from its input(s) the function that it is labeled with --- Boolean AND ($\&$), OR ($\vee$) or NOT ($\neg$).
%The value of $m$ output nodes is the result of the computation.

Each function $f:\{0, 1\}^n\rightarrow\{0, 1\}^m$ can be represented by a Boolean circuit in (infinitely many) different ways.
The complexity of this function $C(f)$ is the size of the smallest circuit that represents this function.
%Thus, the complexity  of a Boolean function is the number of output variables plus the number of gates.

We will also need a formula for the upper bound of the number of different Boolean circuits with a given complexity $C$.
Denote $N(n, m, C)$ to be the number of circuits with $n$ input variables,
$m$ output variables and no more than $C$ gates, that correspond to different Boolean
functions. Then:
\begin{theorem}
\label{shemuSkaits}
$$N(n, m, C)\leq 9^{C+n}(C+n)^{C+m}$$
\end{theorem}
\begin{proof}
Assign to inputs numbers from 1 to $n$, and numbers from $n+1$ to $n+C$ to the gates. Each gate is characterised
with its two inputs (at most $(n+C)^2$ possibilities) and type (AND, OR, NOT, 3 possibilities). There are no more
than $(n+C)^m$ ways how to assign outputs of the circuit and each circuit is counted $C!$ times, one for each
numbering of gates. Therefore the total number of circuits can be estimated as:
$$N(n, m, C)<\frac{(3(C+n)^2)^C\cdot (C+n)^m}{C!}<9^C(1+\frac{n}{C})^C(C+n)^{C+m},$$
here we have used, that $C!>C^C/3^C$ for all $C$.

Further, as $(1+1/x)^x<e<9$ for arbitrary $x>0$, then
$$(1+\frac{n}{C})^C=((1+\frac{n}{C})^\frac{C}{n})^n<9^n$$
from where the result follows.
\qed
\end{proof}

A classical result about Boolean functions states that most of functions $f:\{0, 1\}^n\rightarrow\{0, 1\}^m$ have
approximately the same circuit complexity which is close to maximum. This property is called Shannon effect.
%The lower bound can be obtained from Theorem~\ref{shemuSkaits}, upper bound is found by a smart
%construction of a circuit for an arbitrary Boolean function.
\begin{theorem}[\cite{L84}]
\label{ShTeorema}
For any Boolean function $f:\{0, 1\}^n\rightarrow \{0, 1\}^m$
$$C(f)\lesssim \frac{m2^n}{n+\log m},$$
For almost all  Boolean functions $f:\{0, 1\}^n\rightarrow \{0, 1\}^m$
$$C(f)\gtrsim \frac{m2^n}{n+\log m}.$$
\end{theorem}
Here and further $\log=\log_2$ and we use the notation
$$f(n)\lesssim g(n) \Leftrightarrow \lim_{n\rightarrow \infty}\frac{f(n)}{g(n)}\leq 1.$$

\section{Encodings and Representations of a DFA}
\label{3nod}
Classical representations of automata are table forms or state transition diagrams. They are essentially
the same, a state diagram can be thought of as a visualization of a table form. Table form lists
the transition function of an automaton as a table where each line corresponds to a pair of state and input letter.
In state transition diagram each state is denoted by a circle and for each
transition $(q,x)\rightarrow q'$ an arrow is drawn from state $q$ to state $q'$ above which letter $x$ is written.

Both of these representations show each state of an automaton separately, therefore with these methods
it is not possible to effectively describe an automaton with a large number of states.

One can encode $s$ states into $\left\lceil\log(s)\right\rceil$ (or more) state bits which can be kept in a
 \textit{state register}. Also, input letters can be encoded as a bit vectors. Every automaton
has infinitely many such encodings.

The transition function in this case will take as an input a state register and
an encoded input letter, and produce a (next) state register.
It is thus a Boolean function and it is natural to represent it with a Boolean circuit.

Another question is how to represent the set of accepting states $\tilde{Q}$. We represent it
by a Boolean circuit implementing its characteristic function.
Therefore a representation of
a DFA will consist of an encoding of its state space and input alphabet and two circuits:
one for its transition function and one for the characteristic function
of the set of accepting states.
We call these circuits \textit{transition circuit} and \textit{acceptance circuit}, respectively.

An encoding $E(X)$ of a set X onto a binary string is an injective mapping $f_X:X\rightarrow \{0, 1\}^{b_X}$
where $b_X$ is the length of the encoding. As the mapping is injective then $b_X\geq\lceil\log{|X|}\rceil$.

An encoding of a DFA consists of an encoding of its input alphabet $f_\Sigma$ and an encoding of the state space $f_Q$
which we call input encoding and state encoding, respectively. Additionally for the state encoding
we ask that the start state is encoded as a string of all zeros $f_Q(q_0) = 0^{b_Q}$.

%\begin{definition}
%An encoding $E(A)$ of a DFA $A$ is a pair of injective mappings $E(A)=(f_\Sigma, f_Q)$ such that 
%$f_\Sigma:\Sigma\rightarrow\{0, 1\}^{b_\Sigma}$ maps the input alphabet to a bit vector of length $b_\Sigma$,
%$f_Q:Q\rightarrow\{0, 1\}^{b_Q}$  maps the state space to a bit vector of length $b_Q$ and $f_Q(q_0)=0^{b_Q}$.
%\end{definition}

\begin{definition}
Let $A=(Q,\Sigma,\delta,q_0,\tilde{Q})$ be a given DFA and $(f_\Sigma, f_Q)$ be its encoding.
A pair of Boolean circuits $(F, G)$ is a representation of $A$ under encoding $(f_\Sigma, f_Q)$ iff
\begin{itemize}
\item $F$ has $b_\Sigma+b_Q$ input variables and $b_Q$ output variables,
\item $G$ has $b_Q$ input variables and one output variable,
\item for all $x\in \Sigma$ and $q\in Q$ if $q'=\delta(x, q)$, then $f_Q(q')=F(f_\Sigma(x),f_Q(q))$,
\item $G(f_Q(q))=1 \iff q\in \tilde{Q}$ for all $q\in Q$.
\end{itemize}
\end{definition}

In other words, transition circuit $F$
reads encoded input $f_\Sigma(x)$ as its first $b_\Sigma$ input bits, encoded state $f_Q(q)$ as following
$b_Q$ input bits and has encoded next state $f_Q(q')$ as its $b_Q$ output bits.
Acceptance circuit $G$ reads encoded state $f_Q(q)$  and
outputs 1 as its only output bit iff $q\in\tilde{Q}$.

%A representation of a DFA $A$ is a pair $(E(A), (F, G))$ which consists of its encoding $E(A)$ and a representation
%of this encoding with  a pair of Boolean circuits $(F, G)$. Sometimes if the encoding is irrelevant we will
%call $(F, G)$ to be the representation $A$.

As noted before minimal values for $b_\Sigma$ and $b_Q$ are $\left\lceil\log(|\Sigma|)\right\rceil$
and $\left\lceil\log(|Q|)\right\rceil$ respectively,
but they can be larger as well.
Whether allowing them to be larger gives a possibility to construct smaller representations of DFAs,
is an interesting open question.

It is natural to encode the state space $Q$ with $|Q|$ lexicographically first bit strings 
of length $\lceil\log|Q|\rceil$, in such a case we will say that the state encoding is minimal.
The notion of minimal input encoding is introduced similarly.
We call an encoding of a DFA \textit{minimal encoding} if both encodings: state and input are minimal.
%Any minimal encoding has the property
%that $b_\Sigma=\left\lceil\log(|\Sigma|)\right\rceil$ and $b_Q=\left\lceil\log(|Q|)\right\rceil$.
%We will say that a representation of a language or an automaton is minimal if it is a representation of
%some of its minimal encoding.

%The condition that start state maps to all zero bit string is not too strict. If there is representation of a DFA
%in an encoding with this condition released (state $q_0$ maps to some arbitrary bit string $b_1b_2,\dots, b_n$) then
%one can easy correct it by adding negations to inputs  and outputs for which $b_i=1$, at most $2n$ negations
%have to be added to circuit $F$ (to inputs and outputs) and at most $n$ negations to circuit $G$ (to inputs only),
%here $n$ is the number of state bits.

\section{BC-complexity}
\label{4nod}
%\subsection{Definition and comparison to the state complexity}
In this section we define the main concept of this paper, BC-complexity of a DFA.
We start from the bottom:
\begin{definition}
BC-complexity of a representation of a DFA $(F,G)$ is the sum of complexities
of its transition circuit and acceptance circuit and the number of state bits:
$$\cBC((F, G)) = C(F)+C(G)+b_Q.$$
\end{definition}

The number of state bits $b_Q$ is included in the definition to avoid situation that an automaton
has a large number of states but zero BC-complexity. It is natural to assume that it costs something
to create a circuit even if it has no gates and this is one of the possibilities how to reflect
this in the definition. Another possibility would be to use the complexity of "wires"
instead of the complexity of gates for the underlying circuits, but we prefer to use the standard
complexity for the circuits.

%\begin{definition}
%BC-complexity of a DFA $A$ under encoding $(f_\Sigma, f_Q)$ is the minimal complexity
%over all its representations $R=(F, G)$:
%$$\cBC(E(A))=\min\{C(R):R \mbox{ represents } E(A)\}.$$
%\end{definition}

\begin{definition}
BC-complexity of a DFA $A$,
$\cBC(A)$, is the minimal BC-complexity of its representations:
$$\cBC(A)=\min\{\cBC((F, G)): (F, G) \mbox{ represents } A\}.$$
\end{definition}

Although the name "circuit complexity" also sounds reasonable,
we use the abbreviation "BC-complexity" to avoid confusion
with the circuit complexity of regular languages.
\begin{definition}
BC-complexity of a regular language $L$ is the minimal BC-complexity of all DFAs
that recognize $L$:
$$\cBC(L)=\min\{\cBC(A):A \mbox{ recognizes } L\}.$$
\end{definition}

First we observe that we can optimize our acceptance circuit by rearranging states. If we encode states
in such a way that all accepting states have smaller index than rejecting states (or vice versa) then the acceptance circuit
can be reduced to a comparison operation whose complexity is not greater than $4n$ where $n$ is the number of state bits.

But this is not the best optimization that can be achieved by rearranging states.
%Further we continue by finding upper and lower bounds for BC-complexity
%compared to the state complexity and showing that these bounds are closely reachable.
For the upper bound in the following Theorem~\ref{PFrobezas} different arrangement is used.

\begin{theorem}
\label{PFrobezas}
If $|\Sigma|=k\geq 2$ then for any DFA $A$ with $s$ states, 
$$\lceil\log(s)\rceil\leq \cBC(A)\lesssim (k-1)s.$$
If $|\Sigma|=1$ then for any DFA $A$ with $s$ states, 
$$\lceil\log(s)\rceil\leq \cBC(A)\lesssim \frac{s}{\log s}.$$

\end{theorem}
\begin{proof}
Lower bound. For any representation $(F, G)$ there are $b_Q\geq \lceil\log{s}\rceil$ state bits,
therefore BC-complexity cannot be smaller than $\lceil\log{s}\rceil$.

For upper bound if we just construct an optimal representation $(F, G)$ under
some arbitrary minimal encoding (with $\lceil\log{s}\rceil$ state bits)
then the BC-complexity of this representation according to Theorem~\ref{ShTeorema} can be estimated as
$$\cBC((F, G))\leq C(F)+C(G)+\lceil\log{s}\rceil \lesssim\frac{ks\lceil\log{s}\rceil}{\log (ks\lceil\log{s}\rceil)}
+\frac{s}{\log s}+ \lceil\log{s}\rceil\lesssim ks$$

To improve the result to $(k-1)s$ we will choose a specific minimal encoding where states are ordered in
a way that for one input letter the corresponding transition function is simple.
Denote $q$ to be the encoding of the current state, $q'$ to be the encoding of the next state and $x$ to be
the encoding of the input letter.
We split the transition circuit $F$ in two parts $F_1$ and $F_2$ where part $F_1$
computes the next state for one specific input letter $a$ and part $F_2$ does it for other $k-1$ input letters.
% (Figure \ref{F1F2}).

%\begin{figure}[htb]
%	\centering
%		\includegraphics[width=0.4\textwidth]{bildes/PFrobeza.eps}
%	\caption{Construction of an optimal transition circuit}
%	\label{F1F2}
%\end{figure}

If we look at the state transition graph for the letter $a$ then it
splits into connected components each of which has the form
$$q_1 \rightarrow q_2 \rightarrow \dots q_{m-1} \rightarrow q_m \rightarrow q_j$$
where $1\leq j \leq m$.
Each such component is uniquely defined with two numbers $m$ (the number of states in it) and $j$
(the length of the "tail"), we call $m$ to be the length of a component.
We order all these components by $m$ and $j$ lexicographically what naturally leads
to the ordering of states. Consider all components with parameters $(m, j)$ and denote by $M=M(m, j)$ the index (encoding)
of the first state of the first such component and by $N=N(m, j)$ the index of the last state of the last such component.

The transition function is $q'=q+1$ except for the last state $q_m$ of each component
for which it is $q'=q-(m-j)$. As each of these components have $m$ states then $q$
corresponds to the last state of some component iff $q+1=M \mod{m}$.

The circuit $F_1$ should compute the following function $q' = F_1(q)$:
\begin{verbatim}
q' = q+1
for all pairs (m, j)
  if M(m, j)<=q<=N(m, j) and q+1 == M(m, j) mod m:
    q' = q-(m-j)
\end{verbatim}

Here $M$ and $N$ are the boundaries within which all components with parameters $(m, j)$ are placed.
It is easy to check that circuits for subtraction $q'=q-(m-j)$ and comparison $(M\leq q\leq N)$ are of size $O(\log{s})$,
for modulo comparison $q+1=M \mod{m}$ it is of size $O(\log{s}^2)$. Therefore
the total size of the circuit $F_1$ is
$K*c\log{s}^2$
where $K$ is the number of different pairs $(m, j)$ that correspond to some components
that are present in the transition graph and $c$
is some constant. We need to estimate the maximum value of $K$.

One can easily see that maximum value of $K$ is obtained when
each component with parameters $(m, j)$ appears exactly once and all the smallest components are used.
Let $u$ be the maximum length of a component (maximal value of $m$) under the condition
that all the possible smallest components are used.
For each $m$ there are $m$ possible different types of components ($1\leq j\leq m$) therefore 
$K\leq u(u+1)/2$. 

On the other hand the total length of all components up to the length $u-1$ should be less than $s$:
$$\sum_{m=1}^{u-1} m^2=\frac{(u-1)u(2u-1)}{6}\leq s$$
whence it follows that $u\leq 2\sqrt[3]{s}$.
Therefore
$$K\leq\frac{u(u+1)}{2}\leq\frac{2\sqrt[3]{s}(2\sqrt[3]{s}+1)}{2}\leq 4s^\frac{2}{3}$$.

The size of the transition circuit $F_2$ for the other $k-1$ input letters can be estimated from Theorem~\ref{ShTeorema}:
$$C(F_2)\lesssim\frac{(k-1)s\lceil\log{s}\rceil}{\log ((k-1)s\lceil\log{s}\rceil)}\lesssim (k-1)s.$$

The size of the acceptance circuit can be estimated (from Theorem~\ref{ShTeorema}) as
$C(G)\lesssim \frac{s}{\log{s}}$.

After reordering of states we also have to ensure that the start state is 0. This
can increase the complexity of both circuits by no more than $3\log{s}$ (it is
necessary to add at most $n$ negations to the input of the transition circuit $F$, the output of the transition
circuit $F$, and the input of the acceptance circuit $G$). There are also
$\log{s}$ state bits which are included in the computation of BC-complexity.
We omit these terms of logarithmic order in the computation of BC-complexity because
asymptotically they are negligible.

The BC-complexity of the automaton therefore can be estimated as
$$\cBC(A)\leq C(F)+C(G)+b_Q\lesssim 4c(\log s)^2s^{2/3}+\frac{s}{\log s} + (k-1)s$$

If $k\geq 2$ then the dominant term of this expression is $(k-1)s$ and $\cBC(A)\lesssim (k-1)s$.
For one letter alphabet the dominant term is $s/\log{s}$, therefore $\cBC(A)\lesssim \frac{s}{\log s}$.
\qed
\end{proof}

%Upper bound (sketch of the proof). Denote $n=\lceil\log{s}\rceil$.
%If we just choose arbitrary minimal encoding (with $n$ state bits)
%and construct a pair of minimal circuits $(F, G)$ representing it,
%then the BC-complexity of this representation according to Theorem~\ref{ShTeorema} can be estimated as
%$$C(F)+C(G)\lesssim\frac{ksn}{\log (ksn)}+\frac{s}{\log s}\lesssim ks$$
%To improve the result to $(k-1)$ we will choose a specific minimal encoding where states are ordered in
%a way that for one input letter the corresponding transition function is simple.
%Choose one arbitrary input letter
%and consider the graph of the transition function for this input letter. It consists of connected components
%of the form $q_1 \rightarrow q_2 \rightarrow \dots q_{m-1} \rightarrow q_m \rightarrow q_j$,
%where $1\leq j \leq m$. These components can be arranged according to parameters $m, j$ and this induces
%a natural arrangement of the states of the DFA. The complexity of the transition function for this arrangement
%of states for this input letter is negligible therefore the total BC-complexity of the automaton decreases
%to $(k-1)s$.
%\qed
%\end{proof}

Consider language $L_n$ in binary alphabet $\Sigma=\{0, 1\}$ such that $x\in L_n$
iff $|x|=k$ and $x_{k-n+1}=1$ (the $n$-th letter from the end is "1"). The state complexity
of this language is $2^n$, one has to remember in a state register the last $n$ input letters. But the BC-complexity of it is $n$.
Circuits $F, G$ that represent the natural encoding of a DFA $A_n$ that recognizes $L_n$ have no gates,
they are shown in figure~\ref{ShiftN}. Therefore the BC-complexity of (the representation $(F, G)$ of) $A_n$
is the number of state bits which is $n$.
This example shows that the lower bound of Theorem~\ref{PFrobezas} is strictly reachable.
\begin{figure}[htb]
	\centering
		\includegraphics[width=0.4\textwidth]{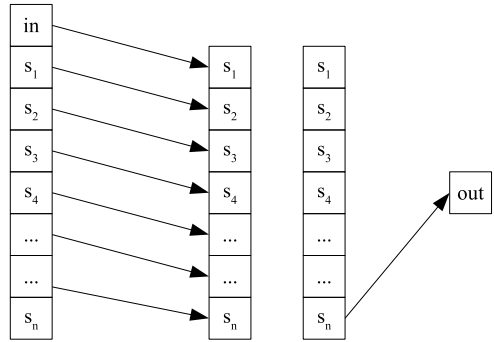}
	\caption{Representation $(F, G)$ of the DFA $A_n$}
	\label{ShiftN}
\end{figure}

Further we try to reach the upper bound. First we find a language
(based on the Shannon function) for which the BC-complexity is at least $s/\log^2(s)$,
afterwards by counting argument we show that BC-complexity for most
languages is close to $(k-1)s$. That matches the upper bound of Theorem~\ref{PFrobezas}
and can be thought of as the Shannon effect for BC-complexity.

Denote by $Sh_n$ the Shannon function on $n$ bits: lexicographically first Boolean function with $n$ input
bits and one output bit with maximal complexity of its minimal circuit.
Consider a language $L^{Sh}_n$ that consists of all words $x_1x_2\dots x_k$ in binary alphabet such that
$Sh_n(x_{k-n+1},x_{k-n+2},\dots,x_k)=1$. 
State complexity of this language is not larger than $2^n$: it is enough to remember the last $n$ input letters.
But its BC-complexity is at least $2^n/{n^2}$.

\begin{theorem}
\label{e2}
BC-complexity of $L^{Sh}_n$ is at least $2^n/n^2$.
\end{theorem}
\begin{proof}
Let $(F, G)$ be a pair of Boolean circuits that represents some DFA $A_n$ recognizing $L^{Sh}_n$.
Assume $F$ has one input bit that represents input letter from the tape and $m=b_Q$ state bits. 
By concatenating $n$ circuits $F$ together with one circuit $G$ as in figure~\ref{ShFpieradijums}.
(state bit output of $j$-th circuit is passed as state bit input
of $j+1$-st) one can obtain a circuit whose size is not larger than $nC(F)+C(G)$ and which
computes Shannon function $Sh_n$ on its $n$ input bits.
From Theorem~\ref{ShTeorema} the complexity of this circuit is at least $2^n/n$.
From $nC(F)+C(G)>2^n/n$ we get that $C(F)+C(G)>2^n/n^2$.
\qed
\end{proof}

\begin{figure}[htb]
	\centering
		\includegraphics[width=0.7\textwidth]{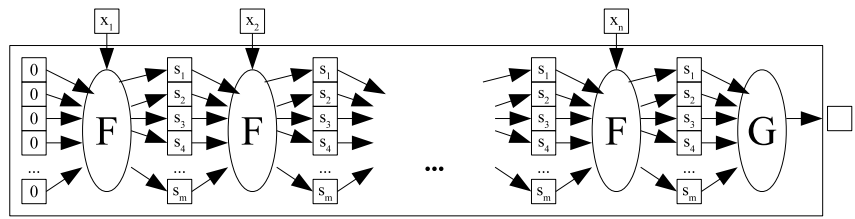}
	\caption{Circuit construction for the Shannon function $Sh_n$}
	\label{ShFpieradijums}
\end{figure}

%\subsection{Shannon effect for BC- and MEBC-complexity}

Theorem~\ref{e2} shows that for some language $L^{Sh}_n$ with $s$ states its BC-complexity is at least $s/(\log{s})^2$.
Next theorem is an extension of this result. With the use of nonconstructive methods (counting argument)
one can show that this value can be raised up to $(k-1)s$. But in the beginning we will need a formula to estimate
the number of automata with a given BC-complexity.

\begin{theorem}
\label{apaksRobMinim}
Fix $\Sigma$ and denote $\mathfrak{A}(c)$ to be the class of those minimal DFAs
whose BC-complexity is less than $c$. If $|\Sigma|=k\geq 2$ then for any $\varepsilon>0$ 
$$\lim_{s\to\infty}\frac{|\mathfrak{A}((1-\varepsilon)(k-1)s)|}{|\mathfrak{A}_s|}=0$$

If $|\Sigma|=1$ then for any $\varepsilon>0$
$$\lim_{s\to\infty}\frac{|\mathfrak{A}((1-\varepsilon)\frac{s}{\log{s}})|}{|\mathfrak{A}_s|}=0$$
\end{theorem}

\begin{proof}
By Theorem~\ref{autSkaits} $\mathfrak{A}_s\geq 2^{s}s^{(k-1)s}$.
Denote $l=2^k$, it is clear that no more than $l$ input bits for data input will be used for the
representation for which BC-complexity is minimal. If more bits are used, then some of them will be equal as there
are only $2^k$ functions that maps $k$ inputs letters to $\{0, 1\} (bits)$.

\begin{figure}[htb]
	\centering
		\includegraphics[width=0.3\textwidth]{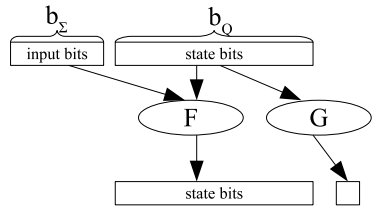}
	\caption{Merged acceptance and transition circuits}
	\label{FGH}
\end{figure}

Consider a representation $(F, G)$ of some encoding $E(A)$ of $A$.
% such that $C(F)+C(G)+b_\Sigma<(1-\varepsilon)(k-1)s$.
Merge these two circuits $F$ and $G$ and obtain one circuit $H$ with $b_Q+b_\Sigma$ inputs and $b_Q+1$ output bits,
the first $b_Q$ of which correspond to the output of the transition circuit $F$, but the last output bit corresponds to the output of the acceptance circuit $G$ (Figure~\ref{FGH}).
The complexity of this circuit $H$ is $C(F)+C(G)$,
for any two minimal automata these "merged" circuits will be different.

Now we want to estimate the number of representations with BC-complexity less that $c$.
Such representations have at least $\lceil\log{s}\rceil$ and no more than $c$ state bits.
The complexity of the "merged" circuit $H$ for a representation with $b_Q$ state bits cannot be more than $c-b_Q$,
the number of such circuits $H$ from theorem~\ref{shemuSkaits} is not larger than
$$N(b_Q+b_\Sigma, b_Q+1, c-b_Q)<N(b_Q+l, b_Q+1, c-b_Q)<9^{c+l}(c+l)^{c+1}.$$

Therefore the number of representations with complexity $c$ is not larger than
$$\sum_{b_Q=1}^c N(b_Q+l,b_Q+1, c-b_Q)<c9^{c+l}(c+l)^{c+1}<9^{c+l}(c+l)^{c+2}.$$

To prove the theorem we have to show that
$$\lim_{s\to\infty}\frac{9^{c+l}(c+l)^{c+2}}{2^ss^{(k-1)s}}=0$$
or what is equivalent to that
$$\lim_{s\to\infty}\log{\left (9^{c+l}(c+l)^{c+2}\right )}-\log{\left (2^ss^{(k-1)s}\right )}=-\infty$$
for the stated values of $c$.

For $k\geq 2$ if we substitute $c=(1-\varepsilon)(k-1)s$ then after simplification we obtain an
equation of the form 
$$\lim_{s\to\infty}-\varepsilon s\log{s}+O(s)=-\infty$$
which is true.
The same happens in the case $k=1$ if we substitute $c=(1-\varepsilon)s/(\log{s})$.
\qed
\end{proof}

We have shown in Theorem~\ref{PFrobezas} that BC-complexity for any regular language with state complexity $s$
and input alphabet of size $k\geq 2$ is not "much larger" than $(k-1)s$. Theorem~\ref{apaksRobMinim} states that for minimal
encodings recognition of almost all such languages would require circuits of size around $(k-1)s$. This can be thought of
as the "Shannon effect" for the BC-complexity of automata: for almost all automata its value is close
to the maximum.

%We do not know
%what is the situation with arbitrary encodings (BC-complexity). We do not think that for a "random" DFA arbitrary encoding
%can be substantially better than minimal, but results of the next section suggest, that in some special cases
%this can be true, although both of these remain open questions.

\section{Minimization of BC-complexity}
\label{6nod}

For the state complexity of DFA an efficient minimization algorithm~\cite{H71} is well known which, given a DFA, finds
the state complexity of it as well as the the minimal DFA itself.
This is in a big contrast with complexity measures of general programs (Turing machines) for which their complexity
(space or time) cannot be determined by any means in the general case.

It is easy to notice that finding the BC-complexity of a DFA is NP-hard.
\begin{theorem}
Finding the BC-complexity of a DFA given its arbitrary representation is NP-hard.
\end{theorem}
\begin{proof}
We will reduce SAT problem to finding minimal BC-complexity of a DFA. Given a SAT problem instance that contains $n$ variables,
consider a DFA with $n$ state bits ($2^n$ states), that works in one letter alphabet, its state transition
function is a "circle", that goes through all the states, and accepting states are those, for which this SAT
instance gives positive output.

Now assume that this SAT instance is not satisfiable --- then this DFA never accepts and therefore its minimal DFA
has 1 state (0 state bits) and its BC-complexity is 0. If this SAT instance
is satisfiable, then any representation of it  will have some state bits and therefore its BC-complexity will be at least 1.
Therefore if one could efficiently find BC-complexity of a given DFA, he could also solve any SAT problem.
\qed
\end{proof}

Further we show one interesting property of BC-complexity --- that
for some DFAs  BC-complexity is significantly smaller than for
their equivalent minimal DFAs. 
The theorem is based on the conjecture that $PSPACE\not\subseteq P/Poly$.
The proof of this theorem for transducers can be found in~\cite{V11}, for DFAs it is almost the same and is omitted here.
Denote by $M(L)$ the minimal DFA recognizing language $L$.

\begin{theorem}
If there is a polynomial $p(x)$ such that
$\cBC(M(L))<p(\cBC(L))$
for all regular languages $L$ then $PSPACE\subseteq P/Poly$.
\label{superTeorem}
\end{theorem}

It means that in some cases by minimizing the number of states
(minimizing state complexity) BC-complexity
of the transition function can increase superpolynomially. And on the other hand, sometimes allowing equivalent states
in the automaton helps to keep BC-complexity small.

\section{BC-complexity applications}
\subsection{Nondeterministic automata}
Theorems~\ref{apaksRobMinim} and \ref{PFrobezas} suggest that for most DFAs in $k$-letter alphabet with $s$ states BC-complexity is around $(k-1)s$. But in many cases when DFAs with a large state space
are constructed by some standard method, it turns out that their BC-complexity is exponentially smaller
than this maximal expected value --- it is of order $Polylog(s)$. Further we look at some of these
standard constructions starting with the determinization of an NFA.

\begin{theorem}
\label{nedetTrans}
If a language $R$ over alphabet $\Sigma$, $|\Sigma|=k$ can be recognized by an NFA $N$ with $n$ states
and $t$ transitions,
then it can also be recognized
by a DFA $A$ for which
$C_{BC}(A)\leq t+(k+1)n+k\log{k}$.
\end{theorem}

\begin{proof}
Consider a DFA $A$ that is obtained by a standard construction from NFA $N$. Its set of states is the
powerset of the set of states of $N$. The state space of $A$ will consist of $2^n$ states
(may be some of them will not be reachable), which can be encoded in $n$ state bits. Each state
bit of an encoding of $A$ will correspond to one state of $N$. For input letters we choose arbitrary
minimal input encoding into $\log{k}$ bits.

The transition circuit of $A$ can be obtained from the transition function of $N$.
NFA $N$ after reading input letter $x\in\Sigma$ will be in state $q_i$,
if there is a state $q_j$, in which it was before (NFA can be in many states simultaneously)
and from which reading input letter $x$ leads to state $q_i$. Denote by $Q_{a}^i$ subset of states of $N$ from which
reading letter $a$ leads to state $q_i$. Denote by $Q_t$ a subset of states in which $N$ is after reading $t$ letters.
If $N$ reads input letter $a$ in step $t$ then:
$$q_i \in Q_{t+1} \leftrightarrow (Q_t\cap Q_{a}^i)\neq\emptyset.$$
In the circuit it means that if $x$ denotes the encoded input letter then
$$q'_i = \bigvee_{a\in\Sigma}( (x=a) \& \bigvee_{q\in Q_{a}^i}q).$$

%Part of the circuit that corresponds to state $q'_1$ in case of $\Sigma=\{0, 1\}$ is shown in figure~\ref{nedetAut}.

To construct all $k$ subcircuits $x=a$ we need $\log{k}$ negations and $k(\log{k}-1)$ conjunctions.

The size of the block $\&\bigvee_{q\in Q_{a}^i}q$ is the number of transitions entering state $q$ on input $a$ therefore the total number of these inner disjunctions and conjunctions for all output bits $q'_i$ is $t$. There are also $(k-1)n$ outer disjunctions
$\bigvee_{a\in \Sigma}$.
In total the complexity of the transition circuit is not larger than $k(\log k-1)+\log{k}+(k-1)n + t \leq t + (k-1)n + k\log{k}$.

Acceptance circuit $G$ is a disjunction of all the final states
of $N$, the complexity of this it is not larger than $n-1$. Also $b_Q=n$ have to be added to the BC-complexity.
Therefore the total BC-complexity of $A$ is not larger than $t+(k+1)n+k\log{k}$.
\qed
\end{proof}

As the number of transitions is not larger than $kn^2$ then
\begin{corollary}
\label{nedet}
If a language $R$ in alphabet $\Sigma$, $|\Sigma|=k$ can be recognized with an NFA $N$ with $n$ states,
then it can also be recognized
with a DFA $A$ for which
$C_{BC}(A)\leq kn^2+(k+1)n+k\log{k}$.
\end{corollary}

%\begin{figure}[htb]
%	\centering
%		\includegraphics[width=0.6\textwidth]{bildes/nedet.eps}
%	\caption{Part of the circuit for DFA $A$ that corresponds to state $q'_1$}
%	\label{nedetAut}
%\end{figure}

%\begin{corollary}
%If $\Sigma=\{0, 1\}$ and language $R$ has an NFA $N$ with $n$ states that recognizes it, then it can be recognized
%also with a DFA $A$, such that $C_{BC}(A)\leq 2n(n+1)$.
%\end{corollary}
%
%This can be viewed other way around as well - if for a language $R$ in alphabet $\Sigma=\{0, 1\}$ $C_{BC}(R)>2n(n+1)$,
%then it cannot be recognized with an NFA with $n$ or less states.

\subsection{Language operations}

State complexity of language operations has been studied long ago, e.g. in~\cite{SY00}. The result of some of the operations
(e.g. reversing) can lead to exponentially larger automata than the original one. Here we analyze how BC-complexity
changes with languages operations and observe that in those cases when the state complexity increases exponentially
it leads to automata whose state transition function is very structured therefore its BC-complexity is exponentially smaller
than state complexity.

For all operations we assume that we are given two languages $L_1$ and $L_2$ and $m=C_s(L_1)$,
$n=C_s(L_2)$, $a=\cBC(L_1)$, $b=\cBC(L_2)$, $k = |\Sigma|$. 
We start with the union and intersection.
\begin{theorem}
If $L_3=L_1\cup L_2$ or $L_3=L_1\cap L_2$ then $\cBC(L_3)\leq a+b+1$.
\end{theorem}
\begin{proof}
Assume circuits $(F_1, G_1)$ represent a DFA recognizing $L_1$ and $(F_2, G_2)$ represent a DFA recognizing $L_2$.
The transition function for a DFA recognizing $L_3$
would consist of circuits $F_1$ and $F_2$ working in parallel. The acceptance circuit consists of circuits $G_1$ and $G_2$
working on corresponding parts of bit vector followed by a disjunction (for union) or conjunction (for intersection) gate.
The number of state bits is the sum of state bits for representations $(F_1, G_1)$ and $(F_2, G_2)$.
The complexity of such a representation is $C(F_1)+C(F_2)+C(G_1)+C(G_2)+1+b_Q=a+b+1$.
\qed
\end{proof}

The complement of the language can be computed by the same pair of circuits as the language itself with
negation added at the end
of the acceptance circuit.
\begin{theorem}
If $L_3=\Sigma^*\setminus L_1$ then $\cBC(L_3)\leq a+1$.
\end{theorem}

A word $x_1x_2\dots x_n$ belongs to the reverse language $L_1^R$ iff $x_n\dots x_2x_1$ belongs to $L_1$.
NFA $N$ recognizing $L_1^R$ can be obtained from the DFA $A$ recognizing $L_1$ by setting the
start state of $N$ to be any accepting state of $A$, setting $q_0$ of $A$ to be the only accepting state of $N$
and reversing all the arrows.
DFA recognizing $L_1^R$ can be obtained from $N$ by running the standard process of determinization. 

%During this process the number of states
%can increase exponentially which happens in some cases~\cite{}.
%In this case also the the BC-complexity can increase exponentially.

\begin{theorem}
$C_{BC}(L_1^R)\leq (2k+1)m+k\log{k}$
\end{theorem}
\begin{proof}
This follows directly from Theorem~\ref{nedetTrans} and the fact, that NFA obtained by reversing all the
transitions has exactly $km$ transitions.
\qed
\end{proof}

Language $L_1L_2$ which is the concatenation of languages $L_1$ and $L_2$ consists of all words $uw$ such that
$u\in L_1$ and $w\in L_2$.
\begin{theorem}
$\cBC(L_1L_2)\leq a+ (2k+1)n+k\log{k}$
\end{theorem}
\begin{proof}
Assume DFA $A_1$ recognizes $L_1$, DFA $A_2$ recognizes $L_2$.
NFA that recognizes $L_1L_2$ can be obtained from $A_1$ and $A_2$  by adding $\varepsilon$-transitions
from all the accepting states of $A_1$ to the start state of $A_2$.
The standard construction of DFA from this NFA can be optimized --- it will consist of circuits $F_1$ and $G_1$ representing
$A_1$ together with a circuit $N(A_2)$ constructed from $A_2$ as from NFA as in Theorem~\ref{nedetTrans}.
Circuit $G_1$ sets state bit corresponding to state $q_0$ of $A_2$ to "1" iff $A_1$ is in accepting state.

By Theorem~\ref{nedetTrans} $C(N(A_2))\leq t + (k+1)n+k\log{k}$ and, since $A_2$ is a deterministic automaton, $t=kn$. Together
it gives that $\cBC(L_1L_2)\leq C(F_1)+C(G_1)+C(N(A_2))\leq a+kn+(k+1)n + k\log{k}=a+(2k+1)n+k\log{k}$.
\qed
\end{proof}

%\begin{figure}[htb]
%	\centering
%		\includegraphics[width=0.6\textwidth]{bildes/LangOpL1L2.eps}
%	\caption{Circuit $F$ representing transition function for a DFA that recognizes $L_1L_2$}
%	\label{L1L2}
%\end{figure}

\begin{theorem}
$\cBC(L_1^*)\leq km^2+(k+1)m+k\log k$.
\end{theorem}
\begin{proof}
NFA recognizing $L_1^*$ can be obtained from DFA recognizing $L_1$ by adding $\varepsilon$-transitions from
all the accepting states to the start state. The resulting NFA therefore also has $m$ states and the result
follows from Corollary~\ref{nedet}.
\qed
\end{proof}
%Table~\ref{valoduOperacijas} compares the state complexity and BC-complexity of language operations.

\begin{table}[htb]
\centering
\begin{tabular}{|c|c|c|}\hline
Operation & State complexity & BC-complexity\\\hline
$L_1 \cup L_2$ & $mn$ & $a+b+1$\\\hline
$L_1 \cap L_2$ & $mn$ & $a+b+1$\\\hline
$\Sigma^* - L_1$ & $m$ & $a+1$\\\hline
$L^R$ & $2^m$ & $ (2k+1)m+k\log{k}$\\\hline
$L_1L_2$ & $(2m-1)2^{n-1}$ & $a+(2k+1)n+k\log{k}$\\\hline
$L_1^*$ & $2^{m-1}+2^{m-2}$ & $ km^2+(k+1)m+k\log{k}$\\\hline
\end{tabular}
\caption{State complexity and BC-complexity of language operations}
\label{valoduOperacijas}
\end{table}

%\section{Connection with Kolmogorov complexity}
%In this section we look at Kolmogorov complexity as an alternative to BC-complexity.
%They aare somewhat similar, both of them tries to describe an automaton in
%a compact way. The main difference is that Kolmogorov complexity shows how easy is to describe
%the automaton, but BC-complexity shows how easy is to execute it on some input.
%
%\begin{definition}
%Kolmogorov complexity of a DFA $A$ is the minimal number $K(A)$ such that there is a Turing
%machine of size $K(A)$ that receives empty tape as input and produces the description (state
%transition table with accepting states marked in it) of $A$.
%\end{definition}
%
%It easy to see that Kolmogorov complexity of a DFA with $n$ states can be arbitrarily small.
%On the other hand BC-complexity is larger than $\log n +1$.
%Another difference is in minimization process. In contrast with BC-complexity which
%can increase superpolynomially when DFA is minimized, Kolmogorov complexity can increase
%only by a fixed constant.
%
%\begin{theorem}
%For all DFA $A$
%$$K(M(A))\leq K(A)+c$$
%for some constant $c$.
%\end{theorem}
%\begin{proof}
%By definition there is a turing machine $N$ of size $K(A)$ such that it produces state transition
%table of automaton $A$. There is a Turing Machine $M$ of size $c$ that given a description
%of a DFA $A$ produces a description of the equivalent minimal automaton $M(A)$. Then
%the turing machine that receives an empty string and first executes $N$ on it
%and then $M$ on the output of $N$ produces the description of $M(A)$ and has
%the size $K(A)+c$.
%\qed
%\end{proof}

\section{Conclusions and open problems}
\label{7nod}
In this paper a new measure of complexity, BC-complexity of DFAs and regular languages, was considered.
Transition function of a DFA as well as the characteristic function of the set of accepting states are expressed
as Boolean circuits and their circuit complexity is taken as a complexity measure (BC-complexity) of this DFA.
It turns out that BC-complexity can vary exponentially for DFA with the same number of states (Theorem~\ref{PFrobezas}).
Theorem~\ref{apaksRobMinim} states that almost all DFAs
BC-complexity is close to maximum ("Shannon effect"). 

In all asymptotic  constructions minimal encodings for state and input alphabet where used,
but it is not known if minimal encodings are always optimal. We think that sometimes they are not, but showing
an example where other encoding
than minimal would be more
efficient (in the sense of minimizing BC-complexity) is an interesting open question.

In section~\ref{6nod} it was shown that BC-complexity of a regular language can be much smaller
than the BC-complexity of the minimal DFA that recognizes it.
%Theorem~\ref{superTeorem} showed a sequence of languages for which BC-complexity is
%more than polynomially smaller than BC-complexity of corresponding minimal DFAs. It means that minimization
%of some DFAs can lead to a significant increase of their BC-complexity.
%Also the sequence of languages in Theorem~\ref{superTeorem} is a little "impractical" from the point of view
%that all of them can accept the language only after exponentially long time - $2^m+n-1$ steps. The motivation for
%introduction of BC-complexity was to have a complexity measure for finite automata which is closer
%to reality than state complexity, so it would be interesting to think of an example with the same
%property, where some word can be accepted already after realistic (polynomial) time.
On the other hand, DFAs with a large state space that are obtained in many standard operations
(determinization of NFA, language operations), have a "good" structure so
that their BC-complexity can be relatively small.

\nocite{*}
\bibliographystyle{eptcs}
\bibliography{MarisValdatsAFL}

\end{document}